\documentclass[11pt,onecolumn,final]{article} 
\usepackage{a4}
\usepackage{fullpage}
\usepackage{expdlist}
\usepackage[latin1]{inputenc}
\usepackage{float}
\usepackage{graphicx}
\usepackage{xspace}
\usepackage{amsmath}
\usepackage{amsfonts}
\usepackage{amssymb}
\usepackage{amsthm}
\usepackage{enumerate}
\usepackage{mathrsfs}

\newcommand{\supcat}{\mathcal{S}} 

\newcommand{\LCP}[0]{\mathsf{LCP}}
\newcommand{\SUF}[0]{\mathsf{SA}}
\newcommand{\SF}[1]{S_{\SUF[#1]..n}}

\newcommand{\psv}[0]{\mathtt{PSV}}
\newcommand{\nsv}[0]{\mathtt{NSV}}
\newcommand{\rmq}[0]{\mathtt{RMQ}}

\newcommand{\lca}{\ensuremath{\textsc{Lca}}}
\newcommand{\sroot}{\ensuremath{\textsc{Root}}}
\newcommand{\ssize}{\ensuremath{\textsc{SubtreeSize}}}
\newcommand{\locate}{\ensuremath{\textsc{LeafLabel}}}
\newcommand{\parent}{\ensuremath{\textsc{Parent}}}
\newcommand{\ancestor}{\ensuremath{\textsc{IsAncestor}}}
\newcommand{\sdepth}{\ensuremath{\textsc{StringDepth}}}
\newcommand{\tdepth}{\ensuremath{\textsc{TreeDepth}}}
\newcommand{\scount}{\ensuremath{\textsc{LeafCount}}}
\newcommand{\fchild}{\ensuremath{\textsc{FirstChild}}}
\newcommand{\nsibling}{\ensuremath{\textsc{NextSibling}}}
\newcommand{\child}{\ensuremath{\textsc{Child}}}
\newcommand{\ithchild}{\ensuremath{\textsc{IthChild}}}
\newcommand{\childrank}{\ensuremath{\textsc{ChildRank}}}
\newcommand{\slink}{\ensuremath{\textsc{SuffixLink}}}
\newcommand{\nrs}{\ensuremath{\textsc{Nrs}}} 
\newcommand{\laq}{\ensuremath{\textsc{Laq}}}
\newcommand{\Null}{\ensuremath{\textsc{null}}}

\theoremstyle{plain}
\newtheorem{definition}{Definition}
\newtheorem{lemma}[definition]{Lemma}

\newcounter{theoremCounter}\newtheorem{theorem}[theoremCounter]{Theorem}

\theoremstyle{remark}

\DeclareMathOperator*{\polyloglog}{polylglg}
\DeclareMathOperator*{\argmin}{argmin}

\title{Combined Data Structure for Previous- and Next-Smaller-Values}

\author{
  Johannes Fischer\thanks{Computer Science Department, Karlsruhe
    University, \texttt{johannes.fischer@kit.edu}}
}

\date{}

\begin{document}
\maketitle

\begin{abstract}
Let $A$ be a static array storing $n$ elements from a totally ordered set. We present a data structure of optimal size at most $n\log_2(3+2\sqrt{2})+o(n)$ bits that allows us to answer the following queries on $A$ in constant time, without accessing $A$: (1) previous smaller value queries, where given an index $i$, we wish to find the first index to the left of $i$ where $A$ is strictly smaller than at $i$, and (2) next smaller value queries, which search to the right of $i$. As an additional bonus, our data structure also allows to answer a third kind of query: given indices $i<j$, find the position of the minimum in $A[i..j]$. Our data structure has direct consequences for the space-efficient storage of suffix trees.
\end{abstract}

\section{Introduction}
We consider the situation where a static array $A[1,n]$ can be preprocessed such that the following three queries can be answered in constant time: previous- and next-smaller-value-queries, where given a position $i$ in $A$, one searches for the next position $p$ to the left (or right) of $i$ with $A[p] < A[i]$, and range minimum queries, where for two given indices $i$ and $j$ we look for the position of the minimum element within the subarray $A[i..j]$.

Our work is situated in the field of succinct data structures, where the aim is to store objects of size $n$ from a universe of size $L(n)$ in $\lg L(n) + (1+o(1))$ bits\footnote{Throughout this article, $\lg$ denotes the binary logarithm.}, while still being able to perform all operations on the data as if they were uncompressed. All succinct data structures work in the word-RAM model of computation, where fundamental operations on a contiguous field of $w$ bits can be performed in constant time ($w$ is the word size, and we assume $\lg n=O(w)$).

Succinct data structures can be further classified into \emph{indexing} and \emph{encoding} data structures. An indexing data structure enhances an object (such as an array) with additional functionality (such as queries) and needs access to the object itself, whereas an encoding data structure recodes all necessary parts of the data for answering the queries without accessing the object.

For range minimum queries alone, there is a data structure in the encoding model of asymptotically optimal size $2n+o(n)$ bits that allows to answer queries in constant time \cite{fischer10optimal}. Previous- and next-smaller-value queries originate from parallel computing \cite{berkman93optimal}. For all three queries combined, the only existing data structure uses $3n+o(n)$ bits \cite{ohlebusch10cst++}.

In this short note, we present an encoding data structure of size at most $n\lg(3+2\sqrt{2})+o(n)\approx 2.54 n+o(n)$ bits that allows to answer all three queries in constant time. It is interesting to note that although we do not have a closed formula for the exact size of our data structure, we prove that it is asymptotically optimal. The reason for this slight oddity is that we are not aware of a closed formula for the size $L$ of the universe of objects that we encode; however, we prove that we encode them in an asymptotically optimal way.

Although our data structure is independent of the underlying array $A$ and may have other applications, our research is clearly motivated by the compact storage of full-text indices \cite{navarro07compressed}. Precisely, we show that our data structure automatically yields the smallest compressed suffix tree with constant-time navigation (we refer the reader to Sect.~\ref{sect:cst} for more details and preliminary work on compressed suffix trees).

The rest of this note is structured as follows. Sect.~\ref{sect:preliminaries} introduces some notation and known results. Sect.~\ref{sect:main} presents the core idea of the paper, a combined data structure for $\rm$s and $\psv$-/$\nsv$-queries. Finally, Sect.~\ref{sect:cst} describes how that data structure yields improvements in compressed suffix trees.

\section{Preliminaries}
\label{sect:preliminaries}
For integers $i$ and $j$, we write $[i,j]$ to denote the set $\{i,i+1,\dots,j\}$, and $(i,j)$ to denote $\{i+1,\dots,j-1\}$. For a rooted tree $\mathcal{T}$ and a node $v$, we write $\mathcal{T}_v$ to denote the subtree of $\mathcal{T}$ rooted at $v$.

\subsection{Queries}
Let $A[1,n]$ be an array of totally ordered objects. For technical reasons, we define $A[0]=-\infty=A[n+1]$ as the ``artificial'' overall minima of the array. We start by formally defining previous smaller values:

\begin{definition}
  \label{def:psv}
  For $1\le i \le n$, let $\psv_A(i)=\max\bigl\{j < i~:~A[j] < A[i]\bigr\}$ denote the \emph{previous smaller value} of position $i$.
\end{definition}

As mentioned in the introduction, we also consider next smaller values and range minima, for completeness formally defined as follows.

\begin{definition}
  \label{def:nsv}
  For $1\le i \le n$, let $\nsv_A(i)=\min\bigl\{j > i~:~A[j] < A[i]\bigr\}$ denote the \emph{next smaller value} of position $i$.
\end{definition}

\begin{definition}
  \label{def:rmq}
  For $1\le i\le j \le n$, let $\rmq_A(i)=\argmin\bigl\{A[k]~:~i\le k\le j\bigr\}$ denote a \emph{range minimum query} between positions $i$ and $j$. If the minimum in the query range is not unique, the leftmost (or rightmost) minimum is chosen as a representative.
\end{definition}

In the following, the subscript $A$ from $\rmq_A$ etc.\ will be omitted if the underlying array $A$ is clear from the context.

\subsection{LRM-Trees }
\label{sect:lrm}

LRM-Trees are the basis of our new data structure. They were introduced under this name as an internal tool for basic navigational operations in ordinal trees~\cite{sadakane10fully}, and, under the name of ``2d-Min Heaps,'' to encode integer arrays in order to support range minimum queries on them~\cite{fischer10optimal}.
 
\begin{definition}[Sadakane and Navarro~\cite{sadakane10fully}; Fischer~\cite{fischer10optimal}]
  \label{def:2dmin}
  The \emph{LRM-Tree} of $A$ is an ordered labeled tree with vertices $0, \dots, n$. For $1 \le i \le n$, $\psv(i)$ is the parent node of $i$. The children are ordered in increasing order from left to right.
\end{definition}

We note the following useful properties of the LRM-Tree (observe that we use nodes and array indices interchangeably throughout this article):

\begin{lemma}[Fischer~\cite{fischer10optimal}]
  \label{lemma:basic}
  Let $\mathcal{T}$ be the LRM-Tree of $A$.
  \begin{enumerate}
  \item The node labels correspond to the preorder-numbers of $\mathcal{T}$ (counting starts at 0).
  \item Let $i$ be a node in $\mathcal{T}$ with children $x_1,\dots,x_k$. Then $A[i] < A[x_j]$ for all $1 \le j \le k$.
  \item Again, let $i$ be a node in $\mathcal{T}$ with children $x_1,\dots,x_k$. Then $A[x_j] \le A[x_{j-1}]$ for all $1 < j \le k$.
  \end{enumerate}
\end{lemma}

\subsection{Succinct Tree Encodings}
\label{sect:succinct_trees}
A rooted ordered tree on $n$ nodes can be encoded in $2n+o(n)$ bits in various ways such that it still permits (the simulation of) all navigational operations in constant time, such as BPS \cite{munro01succinct} or DFUDS \cite{benoit05representing}. Of particular importance to this article are methods based on tree covering (TC) \cite{geary06succinct,he07succinct,farzan08uniform}. They support most navigational operations on trees in constant time, among others $\sroot()$, $\parent(u)$, $\fchild(u)$, $\nsibling(u)$, $\ssize(u)$, selecting the $i$'th child ($\ithchild(u,i)$), computing the rank of a child among its siblings ($\childrank(u)$), and computing lowest common ancestors ($\lca(u,v)$). Farzan and Munro's approach \cite{farzan08uniform} has the further advantage that it can also optimally encode other types of trees, such as those described in the following section.

\subsection{Schr\"oder Trees}
\label{sect:schroeder}
The term \emph{Schr\"oder Tree} is used for various types of rooted ordered trees \cite{stanley99enumerative}: trees with no nodes of out-degree 1, trees with labeled edges, or trees with labeled nodes. For our purposes, we define them as follows.

\begin{definition}
  A \emph{Schr\"oder Tree} is a rooted ordered tree, where any node except the first child in a list of siblings may be colored red or blue. First children are always colored blue.
\end{definition}

The number of Schr\"oder Trees on $n$ nodes is counted by the little Schr\"oder numbers $\supcat_n$. Although we do not have a closed formula for $\supcat_n$, it is known \cite{merlini04waiting} that $\supcat_n=\frac{\rho^n}{\sqrt{\pi n}(2n-1)} (1+O(n^{-1}))$ with $\rho:=3+2\sqrt{2}$. In particular, $\supcat_n \le \rho^n$.

\section{Data Structure}
\label{sect:main}
In this section, we present the new data structure for answering $\rmq$/$\psv$/$\nsv$ on an input array $A$. We start by introducing the general ideas behind our data structure, and then show how this data structure can be encoded succinctly.

\subsection{Basic Solution}
\label{sect:basic}
The LRM-Tree (Def.~\ref{def:2dmin}) encodes all information for answering $\psv$-queries in a natural way, as it suffices to move to the parent node of $i$ for answering $\psv(i)$. It also captures all sufficient information for answering $\rmq$s:

\begin{lemma}[Fischer~\cite{fischer10optimal}]
  \label{lemma:rmqlca}
  For arbitrary nodes $i$ and $j$ in the LRM-Tree of $A$, $1 \le i < j \le n$, let $\ell=\lca(i,j)$. Then if $\ell=i$, $\rmq(i,j)$ is given by $i$, and otherwise, $\rmq(i,j)$ is given by the child of $\ell$ that is on the path from $\ell$ to $j$.
\end{lemma}

Thus, it remains to show how $\nsv$-queries can be answered. It is easy to see that the LRM-Tree alone is not enough for this task: consider $A=[0,0]$ and $A'=[1,0]$. These arrays have the same LRM-Tree (and hence the same answers to all $\rmq$s and $\psv$-queries); yet, their $\nsv$-queries differ, as $\nsv_A(1)=3$, and $\nsv_{A'}(1)=2$.

In principle, we could build another LRM-Tree $\mathcal{T}^\textrm{R}$ on the reversed sequence $A^\textrm{R}$ for answering $\nsv$=queries, as $\nsv_A(i)=n-\psv_{A^\textrm{R}}(n-i+1)+1$. As this would double the space of the resulting data structure, we now present a more sophisticated solution.

The general idea of our data structure can be seen as follows. Recall property 3 of Lemma~\ref{lemma:basic}: the children $x_1,\dots,x_k$ of a node $v$ in the LRM-Tree are ordered decreasingly from left to right: $A[x_1] \ge A[x_2] \ge \dots \ge A[x_k]$. Now suppose we wish to calculate $\nsv(x_i)$ for some $1\le i \le k-1$, and assume that $A[x_i] > A[x_{i+1}]$. Then $\nsv(x_i)=x_{i+1}$, as all $A$-values in the subtree $\mathcal{T}_{x_i}$ are strictly greater than at position $x_i$ (property 2 of Lemma~\ref{lemma:basic}). If, on the other hand, $A[x_i] = A[x_{i+1}]$, then the next ``candidate'' for $\nsv(x_i)$ is $x_{i+2}$ (assuming $i\le k-2$), as again all $A$-values in $\mathcal{T}_{x_{i+1}}$ are strictly greater than $A[x_{i+1}]=A[x_i]$.

This suggests the following general approach. In the LRM-Tree $\mathcal{T}$ of $A$, a node is colored \emph{red} if the corresponding value in $A$ is smaller than the $A$-value at its left sibling (if such a sibling exists). More formally, let $v$ be a node in $\mathcal{T}$ with children $x_1,\dots,x_k$. Then for all $2 \le i \le k$, node $x_i$ is colored red if and only if $A[x_i] < A[x_{i-1}]$. All other nodes (including the root) are colored blue. We call the resulting tree the \emph{Colored LRM-Tree}.

To get the connection to $\nsv$-queries, we need the following definition:
\begin{definition}
  \label{def:nrs}
  Let $\mathcal{T}^\textrm{C}$ the Colored LRM-Tree of $A[1,n]$, and let $v$ be a node in $\mathcal{T}^\textrm{C}$ with children $x_1,\dots,x_k$. The \emph{next red sibling} $\nrs(x_i)$ of a node $x_i$ is the leftmost sibling to the right of $x_i$ that is colored red. If such a sibling does not exist, we define $\nrs(x_i)=\perp$. In symbols, let $M = \{i<j\le k~:~x_j\text{~is colored red}\}$. Then $\nrs(x_i)=\perp$ if $M=\emptyset$, and otherwise $\nrs(x_i)=x_{\min M}$.
\end{definition}

We can then show the following lemma:

\begin{lemma}
  \label{lemma:nsv}
  Let $\mathcal{T}^\textrm{C}$ the Colored LRM-Tree of $A[1,n]$, and let $v$ be a node in $\mathcal{T}^\textrm{C}$ with children $x_1,\dots,x_k$, $x_1\le x_2 \le \dots \le x_k$. Then
  $$
  \nsv(x_i)=\begin{cases}
    \nrs(x_i)  & \text{if }\nrs(x_i)\ne\perp\\
    x_k + \ssize({x_k}) & \text{otherwise.}
  \end{cases}
  $$
\end{lemma}
\begin{proof}
  We consider each case in turn.
  \begin{description}[\setlabelstyle{\normalfont}]
  \item[$\nrs(x_i)\ne\perp$.]
    Let $j$ be defined by $x_j=\nrs(x_i)$. From Def.~\ref{def:nrs} and the fact that node $x_j$ is red, we know that $A[x_j]<A[x_i]$. Hence, we need to show that $A[h]\ge A[x_i]$ for all $h \in (x_i,x_j)$. From property 1 of Lemma~\ref{lemma:basic}, we know that all values in $(x_i,x_j)$ occur in $\mathcal{T}^\textrm{C}_{x_i}, \dots, \mathcal{T}^\textrm{C}_{x_j-1}$. Because $j$ is minimal and due to property 3 of Lemma~\ref{lemma:basic}, $A[h]=A[x_i]$ for $h=x_{i+1},\dots,x_{j-1}$. But due to property 2 of Lemma~\ref{lemma:basic}, $A[h]>A[x_i]$ for all $h\in[x_{\alpha}+1,x_{\alpha+1}-1]$ and all $i\le\alpha\le j-1$. Hence, $\nsv(x_i)=x_j$.
  \item[$\nrs(x_i)=\perp$.]
    Let $y = x_k + \ssize({x_k})$. As above, we can show that $A[h]\ge A[x_i]$ for all $x_i < h < y$. It thus remains to show that $A[y] < A[x_i]$. For the sake of contradiction, assume that $A[y] \ge A[x_i]$, where we further distinguish between the cases ``$=$'' and ``$>$.'' If $A[y] = A[x_i]$, then $\psv(y) = v$ (the parent node of $x_i$), so $y$ is the right sibling of $x_k$, a contradiction to the definition of $x_k$. If $A[y] > A[x_i]=A[x_k]$, then again due to property 2 of Lemma~\ref{lemma:basic}, we have $\psv(y)\in[x_k,y-1]$. So $\mathcal{T}^\textrm{C}_{x_k}$ contains $y$, a contradiction to the size of $\mathcal{T}^\textrm{C}_{x_k}$, which is $y-x_i$, as $\mathcal{T}^\textrm{C}_{x_k}$ contains exactly those elements from $[x_k,y)$.
  \end{description}
\end{proof}

\subsection{Succinct Encoding}
\label{sect:succinct}
We represent the Colored LRM-Tree $\mathcal{T}^\textrm{C}$ from Sect.~\ref{sect:basic} similar to Farzan and Munro's succinct TC-encoding for ordinal trees \cite{farzan08uniform}. This approach is based on a two-level decomposition of the tree into mini- and micro-trees. In our scenario, the encoding of a micro-tree is simply its index in an enumeration of all Schr\"oder Trees of the micro-tree size (called ``enumeration code'' in \cite{farzan08uniform}). In total, this uses optimal $\lg\supcat_n+o(n)$ bits of space.

It remains to show how we implement the query algorithms for $\rmq$, $\psv$, and $\nsv$.

As $\psv(i)=\parent(i)$ and the parent-operation is directly supported by TC, we can directly focus on $\nsv$. Recall Lemma \ref{lemma:nsv}: given $i$, we need to find $\nrs(i)$ in order to answer $\nsv(i)$. The $\nrs$-method can be implemented as the combination of \emph{modified} $\ithchild$- and $\childrank$-operations, as they are described by Farzan and Munro \cite{farzan08uniform} (see \cite[p.~23]{farzan09phd} for further details). In particular, given node $i$, we find the parent $p$ of $i$, and then determine the rank $r$ of $i$ among all its red siblings, from where we select the $r+1$'st red node. To this end, if $p$ is a root of a mini- or micro-tree, we use a \emph{modified} fully indexable dictionary (FID) \cite{raman02succinct} to rank/select among the red nodes. These FIDs are similar to the ones already stored at each mini- or micro-tree root, with the difference that they index only the red nodes. Similar to the original analysis, their overall space contributes only $o(n)$ bits to the final space. If, on the other hand, $p$ is not a mini- or micro-tree root, we use the lookup-tables stored along with the micro-trees to rank/select among the red nodes. These lookup-tables also use only $o(n)$ bits, as we use micro-trees of size $O(\log_\rho n/4)$. Finally, if $\nrs(i)=\perp$, we move to the rightmost sibling $j$ of $i$ and count the subtree size at $j$; both operations are supported in $O(1)$ time by TC.

For implementing $\rmq(i,j)$, we have to show how the operations in Lemma~\ref{lemma:rmqlca} can be performed in constant time. We cannot resort to the method described by Fischer \cite{fischer10optimal}, as it is inherently connected to DFUDS. We thus do the following: first compute $\ell=\lca(i,j)$; this is supported by TC \cite{geary06succinct,he07succinct}. Then if $\ell\ne i$ (otherwise we return $i$), compute the depth $d$ of $\ell$ (depth is supported by TC). Finally, compute the child of $\ell$ that is on the path to $j$ by a level-ancestor query $\laq(j,d+1)$ (supported by TC); this is the answer.

\begin{theorem}
  \label{thm:main1}
  For an array of $n$ totally ordered objects, there is a data structure using $\lg \supcat_n + o(n) \le n\lg(3+2\sqrt{2})+o(n) \approx 2.54n+o(n)$ bits of space that supports $\rmq$s, $\psv$- and $\nsv$-queries on $A$ in $O(1)$ time, without accessing $A$ at query time.
\end{theorem}

\subsection{Optimality}
It is easy to see that the encoding from Sect.~\ref{sect:succinct} is optimal. Given any data structure $\mathscr{D}_A$ supporting $\psv$ and $\nsv$ on some underlying input array $A$, we can reconstruct the Colored LRM-Tree $\mathcal{T}^\textrm{C}$ of $A$, without knowing $A$: We first create $\mathcal{T}^\textrm{C}$'s rightmost path $n=x_1, x_2, \dots, x_k=0$ in a bottom-up manner, by successively querying $x_{i+1} = \psv(x_i)$, until arriving at $x_k=0$. All nodes are initially colored blue. This leaves us with unprocessed intervals $[x_i+1,x_{i+1}-1]$, which are handled recursively. During these recursive calls, suppose that a query $\psv(v)$ brings us to a node $u$ which is already present in the (partial) LRM-Tree $\mathcal{T}^\textrm{C}$. Let $w$ be the smallest child of $u$ greater than $v$ (i.e., the leftmost child of $u$ to the right of $v$). We then check if $\nsv(v)=w$, in which case we color $w$ red. Otherwise ($\nsv(v)> w$), $w$ remains blue, as in this case $A[v] = A[w]$. This procedure correctly reconstructs the Colored LRM-Tree $\mathcal{T}^\textrm{C}$ of $A$.

As every Schr\"oder Tree is also a Colored LRM-Tree for some array $A$ (starting at the root with children $x_1,\dots,x_k$, set $A[x_k]$ to 0, and $A[x_{i-1}]$ to $A[x_{i}]$ or $A[x_{i}]+1$, depending on whether $x_i$ is colored blue or red; the unprocessed intervals are handled recursively), we need at least $\lg \supcat_n$ bits to encode $\mathscr{D}_A$ in the worst case. This proves the optimality of the data structure from Thm.~\ref{thm:main1}.

\section{Application to Compressed Suffix Trees}
\label{sect:cst}
The result from Thm.~\ref{thm:main1} has direct consequences for compressed suffix trees (CSTs). A suffix tree (ST) for a string $S$ of length $n$ is a compact trie storing all the suffixes of $S$, in the sense that the characters on any root-to-leaf path spell out exactly a suffix. The ST is an extremely important data structure with applications in exact or approximate string matching, bioinformatics, and document retrieval, to mention only a few examples.

The drawback of STs is their huge space consumption of 20--40 times the text size ($O(n\lg n)$ bits in theory), even when using carefully engineered implementations. To reduce their size, in recent years several authors provided compressed variants of STs \cite{munro01space,grossi05compressed,sadakane07compressed,russo08fully,fischer09faster,ohlebusch09compressed,canovas10practical,ohlebusch10cst++,fischer10wee}.

We regard the CST as an abstract data type supporting the following operations (apart from the usual navigational operations on trees as those mentioned in Sect.~\ref{sect:succinct_trees}): $\scount(u)$ gives the number of leaves (suffixes) below $u$, $\locate(u)$ for a leaf $u$ yields the position in $S$ where the corresponding suffix begins, $\sdepth(u)$ gives $u$'s string-depth (number of characters on the root-to-$u$ path), $\slink(u)$ gives the unique node $v$ with root-to-$v$ label $\alpha \in \Sigma^\star$ if the root-to-$u$ label is $a\alpha$ for some $a \in \Sigma$, and $\child(u,a)$ gives the child $v$ of $u$ such that the label on the edge $(u,v)$ starts with $a\in\Sigma$. Here and in the following, $\Sigma$ denotes the underlying alphabet of size $\sigma$. See the first column of Tbl.~\ref{tbl:cst} for all operations (\emph{level ancestor queries} are excluded as we are not a aware of any actual algorithm that needs them in a suffix tree).

A CST on $S$ can be divided into three components: (1) the \emph{suffix array} $\SUF$, specifying the lexicographic order of $S$'s suffixes, defined by $\SF{1} < \SF{2} < \cdots < \SF{n}$ (hence $\SUF$ captures information on the \emph{leaves}); (2) the \emph{LCP-array} $\LCP$, storing the lengths of the \emph{longest common prefixes} of lexicographically adjacent suffixes: $\LCP[1] = -1$ and for $2\leq i \le n$, $\LCP[i] = \max\{k \ge 0~:~S_{\SUF[i]..\SUF[i]+k-1}=S_{\SUF[i-1]..\SUF[i-1]+k-1}\}$, which is the string-depth of the LCA of the lexicographically $i$'th and $i-1$'st suffix (hence $\LCP$ captures information on \emph{internal} nodes); and (3) additional data structures for simulating the \emph{navigational operations}. The goal of a CST is to compress each of these three components.

\begin{table}[t]
  \small
  \centering
  \caption{Comparison of different CSTs (space in bits on top of $\SUF$ and $\LCP$). The $O(\cdot)$ is omitted in all operations. Trees \cite{russo08fully,fischer09faster,ohlebusch09compressed} are incomparable to our approach, as they use less space in exchange for higher navigation times. $t_\psi$ denotes the time to compute the position of $\SUF[\cdot]+1$ in $\SUF$, which is $O(1)$ in most compressed suffix arrays.}
  \label{tbl:cst}
  \begin{tabular}{|l||c|c||c|c|c|c|}
    \hline
    & \textbf{\cite{russo08fully}} & \textbf{\cite{fischer09faster,canovas10practical}}& \textbf{\cite{munro01space,grossi05compressed,sadakane07compressed}} & \textbf{\cite{ohlebusch09compressed}} & \textbf{\cite{ohlebusch10cst++}} & \textbf{NEW}\\
    \hfill space & $o(n)$ & $o(n)$ & $4n$ & $2n$ & $3n$ & \boldmath$2.54n$ \\ \hline\hline
    \sroot & 1 & 1 & 1 & 1 & 1 & \textbf{1}\\\hline
    \ancestor & 1 & 1 & 1 & 1 & 1 & \textbf{1}\\\hline
    \ssize & --- & --- & 1 & --- & --- & \textbf{---}\\\hline
    \scount & 1 & 1 & 1 & 1 & 1 & \textbf{1}\\\hline
    \locate & $\lg^{1+\alpha}n$ & $t_\SUF$ & $t_\SUF$ & $t_\SUF$ & $t_\SUF$ & \boldmath$t_\SUF$\\\hline
    \sdepth & $\lg^{1+\alpha} n$ & $t_\LCP$ & $t_\LCP$ & $t_\LCP$ & $t_\LCP$ & \boldmath$t_\LCP$ \\\hline
    \parent & $\lg^{1+\alpha} n$ & $t_\LCP\polyloglog n$ & 1 & $t_\LCP\lg\sigma$ & 1 & \textbf{1} \\\hline
    \fchild & $\lg^{1+\alpha} n$ & $t_\LCP\polyloglog n$ & 1 & $t_\LCP$ & 1 & \textbf{1} \\\hline
    \nsibling& $\lg^{1+\alpha} n$ & $t_\LCP\polyloglog n$ & 1 & $t_\LCP$ & 1 & \textbf{1} \\\hline
    \slink& $\lg^{1+\alpha} n$ & $t_\psi + t_\LCP\polyloglog n$ & $t_\psi$ & $t_\psi + t_\LCP\lg\sigma$ & $t_\psi$ & \boldmath$t_\psi$\\\hline
    \lca & $\lg^{1+\alpha} n$ & $t_\LCP\polyloglog n$ & 1 & $t_\LCP\lg\sigma$ & 1 & \textbf{1}\\\hline
    \tdepth & $\lg^{2+2\alpha} n$ & --- & 1 & --- & --- & \textbf{---} \\\hline
    \child & $\lg\sigma+\lg^{1+\alpha} n$ & \scriptsize $t_\LCP\polyloglog n + t_\SUF \lg\sigma$ & $t_\SUF \lg\sigma$ & $t_\SUF \lg\sigma$ & $t_\SUF \lg\sigma$ & \boldmath$t_\SUF \lg\sigma$\\\hline
  \end{tabular}
\end{table}

We do not discuss here the different time/space tradeoffs for compressing $\SUF$ and $\LCP$; we just mention that both can be compressed into space proportional to the entropy of the underlying text, at the cost of increased access times, which we denote by $t_\SUF$ and $t_\LCP$, respectively.

Of more interest to us is the fact that most recent CSTs \cite{fischer09faster,ohlebusch09compressed,ohlebusch10cst++} represent a node $v$ as an interval $[v_l:v_r]$ in $\LCP$ and base their navigation on $\rmq$s and $\psv$-/$\nsv$-queries in $\LCP$. There are two basic strategies for supporting these queries: we can either use structures of size $o(n)$ \cite{fischer09faster,canovas10practical} or $2n+o(n)$ \cite{ohlebusch09compressed} bits and substitute ``missing information'' by a (sub-)logarithmic number of lookups to $\LCP$ (indexing model), resulting in increased navigation time (see 3rd and 5th column Tbl.~\ref{tbl:cst}). The other option \cite{ohlebusch10cst++} is to use a data structure that computes $\rmq$/$\psv$/$\nsv$ without needing access to the underlying LCP-array (encoding model).

Given these observations, the index from Thm.~\ref{thm:main1} almost directly yields a CST with $\approx 2.54n + o(n)$ bits on top of $\SUF$ and $\LCP$ with constant-time support of all operations that do not necessarily need access to $\SUF$ or $\LCP$. See again Tbl.~\ref{tbl:cst} for a comparison. In particular, we get the smallest CST with constant-time navigation. Note that it is of utmost theoretical and practical importance to have the smallest possible data structure for the navigational component of a CST, as its $O(n)$-term is incompressible, whereas the space of the other two components of a CST ($\SUF$ and $\LCP$) vanishes if the entropy of the underlying text does.

All suffix tree operations (apart from $\scount$, $\sdepth$, and $\child$) from Tbl.~\ref{tbl:cst} can be implemented solely by performing $\rmq$s and $\psv$-/$\nsv$-queries in $\LCP$, see \cite{fischer09faster,ohlebusch10cst++}. Only the implementation of the $\nsibling$-operation relies on structures that are proprietary to \cite{ohlebusch10cst++} (and the one in \cite{fischer09faster} accesses $\LCP$); we therefore give our own implementation as follows: let $v=[v_l:v_r]$ be the node whose next sibling we want to compute. First check if $v$ equals the root, and return $\Null$ in this case. Otherwise, compute $w=[w_l:w_r]=\parent(v)$. If $v_r=w_r$, return $\Null$, as $v$ does not have a right sibling in this case. We now know that $v_r+1$ is the leftmost index of $\nsibling(v)$. To determine the rightmost index, check if $\nsv(v_r+1)=w_r+1$, and return $[v_r+1,w_r]$ in this case, as then $v$ is the second-to-last child of $w$. Otherwise, return $[v_r+1,\rmq(v_r+2,w_r)-1]$, as the range minimum query returns a position in $\LCP$ where the string-depth of $w$ is stored.

\begin{theorem}
  \label{thm:main2}
  Let $S$ be a text of size $n$ with characters from an alphabet of size $\sigma$. Given $S$'s suffix array with access time $t_\SUF$ and its LCP-array with access time $t_\LCP$, there is a CST with additional $\lg \supcat_n + o(n)\le 2.54n+o(n)$ bits that supports the operations as indicated in the last column of Tbl.~\ref{tbl:cst}.\hfill \qed
\end{theorem}

Our CST resides in between \cite{ohlebusch09compressed} and \cite{ohlebusch10cst++}: it is smaller than \cite{ohlebusch10cst++} and larger than \cite{ohlebusch09compressed}, but equally fast as the larger of these \cite{ohlebusch10cst++}.

It is interesting to note that our $\lg\supcat_n\approx 2.54n$ bits are also optimal for encoding the topology of a suffix tree, as it is a tree with exactly $n$ leaves and no nodes of out-degree 1; the number of such trees is also counted by the little Schr\"oder number $\supcat_n$. However, we cannot make an optimality claim for the CST from Thm.~\ref{thm:main2}, as it builds on $\SUF$ and $\LCP$, who already capture the topology of the suffix tree.

\bibliographystyle{abbrv}
\bibliography{paper}

\end{document}